\documentclass[nofootinbib,twocolumn,pra,superscriptaddress]{revtex4-2}
\usepackage{hyperref}
\usepackage{latexsym,amsmath,amssymb,amsfonts,graphicx,color,amsthm}
\usepackage{enumerate}
\usepackage{braket}
\usepackage{adjustbox}
\usepackage{footnote}
\usepackage[dvipsnames]{xcolor}
\usepackage{tikz}
\usepackage{subcaption}

\newcommand{\cG}{\mathcal{G}}
\newcommand{\cH}{\mathcal{H}}
\newcommand{\cI}{\mathcal{I}}
\newcommand{\cJ}{\mathcal{J}}
\newcommand{\cK}{\mathcal{K}}
\newcommand{\cL}{\mathcal{L}}

\newcommand{\cS}{\mathcal{S}}

\def \diracspacing {0.7pt}
\newcommand{\ketbra}[2]{| \hspace{\diracspacing} #1 \rangle \langle #2 \hspace{\diracspacing} |} 
\newcommand{\ketbraq}[1]{\ketbra{#1}{#1}} 

\newcommand{\Id}{\mathbb{I}}
\newcommand{\I}{\mathbb{I}}

\DeclareMathOperator{\tr}{tr}

\newtheorem{theorem}{Theorem}
\newtheorem{proposition}{Proposition}

\newtheorem{definition}{Definition}

\newtheorem{example}{Example}

\newtheorem{observation}{Observation}

\begin{document}

\title{On the compatibility of quantum instruments}
\author{Arindam Mitra}
\email{amitra@imsc.res.in}
\affiliation{Optics and Quantum Information Group, The Institute of Mathematical Sciences,
C. I. T. Campus, Taramani, Chennai 600113, India}
\affiliation{Homi Bhabha National Institute, Training School Complex, Anushaktinagar, Mumbai 400094, India}

\author{Máté Farkas}
\email{mate.farkas@icfo.eu}
\affiliation{ICFO-Institut de Ciencies Fotoniques, The Barcelona Institute of Science and Technology, 08860 Castelldefels, Spain}

\date{\today}

\begin{abstract}

Incompatibility of quantum devices is a useful resource in various quantum information theoretical tasks, and it is at the heart of some fundamental features of quantum theory. While the incompatibility of measurements and quantum channels is well-studied, the incompatibility of quantum instruments has not been explored in much detail. In this work, we revise a notion of instrument compatibility introduced in the literature that we call traditional compatibility. Then, we introduce the new notion of parallel compatibility, and show that these two notions are inequivalent. Then, we argue that the notion of traditional compatibility is incomplete, and prove that while parallel compatibility captures measurement and channel compatibility, traditional compatibility does not. Hence, we propose parallel compatibility as the conceptually complete definition of compatibility of quantum instruments.
\end{abstract}

\maketitle

\section{Introduction}

Incompatibility (the lack of compatibility) is one of the features of quantum theory that sets it apart from classical physics \cite{heino-review}. Intuitively, two quantum devices are compatible if there exists a joint device such that implementing the joint device is equivalent to simultaneously implementing the two original devices. While incompatibility may at first sound like a drawback, in fact the incompatibility of quantum measurements leads to practical advantages in various quantum information processing tasks \cite{wootters-Field,Cerf,Heino-qrac,Uola-steering}. From the foundational point of view, the incompatibility of quantum channels is intimately linked to the well-known no-cloning theorem \cite{Heino-Incompa-chan,Wootters}, and the incompatibility of the identity channel and a non-trivial measurement is linked to the uncertainty principle \cite{Heino-obs-chan}.

While the incompatibility of measurements and quantum channels is well-studied, much less effort has been designated to the study of the incompatibility of quantum \textit{instruments}, a more general class of quantum devices capturing measurement processes in their full detail. In this work, we review a definition of instrument compatibility used in the literature (which we call \textit{traditional} compatibility), and address its conceptual adequacy and its relation to measurement and channel compatibility. We then define a new notion of instrument compatibility (which we call \textit{parallel} compatibility) and argue that this notion is conceptually more in line with the well-established notions of measurement and channel compatibility. We further prove that parallel compatibility captures measurement and channel compatibility in a well-defined manner, while traditional compatibility cannot capture channel compatibility. We therefore propose to adapt the notion of parallel compatibility instead of traditional compatibility.

The rest of this paper is organised as follows. In Sec.~\ref{Sec:preli}, we discuss the mathematical background, various quantum devices and their compatibility. In Sec.~\ref{sec:main}, we discuss the compatibility of instruments. In particular, in Sec.~\ref{subsec:main_1}, we review the definition of traditional compatibility, introduce the concept of parallel compatibility, and show that these two notions are conceptually different. In Sec.~\ref{subsec:main_2}, we argue that the traditional definition of compatibility of instruments is conceptually incomplete---unlike parallel compatibility---and show that it cannot capture channel compatibility. In Sec.~\ref{subsec:main_3}, we prove that parallel compatibility of quantum instruments can capture the idea of measurement compatibility, channel compatibility and measurement-channel compatibility. We conclude in Sec.~\ref{sec:conc}, and lay out potential connections of parallel compatibility to certain information-theoretic tasks.

\section{Preliminaries}\label{Sec:preli}

In this section, we discuss the mathematical background, different types of quantum devices and their compatibility. In general, in quantum theory to every physical system there is an associated Hilbert space, $\cH$, which we assume to be finite dimensional. The set of linear operators on $\cH$ is denoted by $\cL(\cH)$, the set of positive linear operators on $\cH$ is denoted by $\cL^+(\cH)$,  and a quantum state is described by a positive semidefinite operator, $\rho \geq 0$, with unit trace. The convex set of all states on $\cH$ is denoted by $\cS(\cH)$.

\subsection{Measurements}

Measurements can be thought of as quantum devices that take a quantum state as an input and produce a classical output (the measurement outcome). Mathematically, measurements are described by positive operator-valued measures (POVMs), which in the $n$-outcome case correspond to a set of $n$ positive semidefinite operators, $A = \{A(x)\}_{x=1}^n$ such that $\sum_x A(x) = \I$, where $\I$ is the identity operator. In the following, we will denote the outcome set of $A$ by $\Omega_A$.

\subsection{Quantum channels}

Quantum channels map quantum states to quantum states, that is, they are devices with a quantum input and a quantum output. Mathematically, quantum channels are described by completely positive trace-preserving (CPTP) maps $\Lambda:\cS(\cH)\rightarrow\cS(\cK)$ \cite{Nielsen,Heino-book}. A useful characterisation of quantum channels is the \textit{Kraus representation}, i.e., any quantum channel $\Lambda$ can be written as $\Lambda(\rho)=\sum_i K_i\rho K^{\dagger}_i$ where the $K_i$'s are linear operators called the Kraus operators, and they satisfy $\sum_iK_i^{\dagger}K_i=\Id$.

The dual of a map $\Lambda: \cS(\cH) \to \cS(\cK)$ is a map $\Lambda^\ast: \cL(\cK) \to \cL(\cH)$ such that $\tr[ \Lambda(X) Y ] = \tr[ X \Lambda^\ast(Y) ]$ for all $X \in \cL(\cH)$ and $Y \in \cL(\cK)$. It is easy to show that the dual of a CPTP map is a completely positive (CP) unital map, i.e., a CP map that maps the identity to the identity. The following observation on dual channels will be useful later:
\begin{observation}\label{obs:marginal_dual}
Let $\Gamma:\cS(\cH)\rightarrow\cL^+(\cK_1\otimes\cK_2)$ be a CP map and  $\Lambda:\cS(\cH)\rightarrow\cL^+(\cK_1)$ be a CP map such that $\Lambda(\rho)=\tr_{\cK_2}[\Gamma(\rho)]$ for all $\rho \in \cS(\cH)$. Then $\tr[\Lambda(\rho)]=\tr[\Gamma(\rho)]$, and thus $\tr[\rho\Lambda^\ast (\Id_{\cK_1})]=\tr[\rho\Gamma^\ast (\Id_{\cK_1\otimes\cK_2})]$ for all $\rho \in \cS(\cH)$. Using the fact that the dual map of a CP map is CP and that one can always find a positive semidefinite basis of Hermitian matrices, we conclude that $\Lambda^\ast (\Id_{\cK_1})=\Gamma^\ast (\Id_{\cK_1\otimes\cK_2})$.
\end{observation}

\subsection{Quantum instruments}

Quantum instruments simultaneously generalise measurements and quantum channels: they take a quantum state as an input and provide both a classical and a quantum output. One may think of a quantum instrument as a measurement process, by associating the classical output with the measurement outcome, and the quantum output with the post-measurement state. Mathematically, a quantum instrument $\cI$ is defined as a set of CP maps $\{\Phi_x:\cS(\cH)\rightarrow\cL^+(\cK)\}$ such that $\Phi^{\cI} \equiv \sum_x\Phi_x$ is a CPTP map \cite{Heino-book}. Given a quantum state $\rho$, the classical output of the instrument is $x$ and the quantum output is $\Phi_x(\rho)$, both with probability $\tr[ \Phi_x(\rho)]$.

Given a measurement $A$, we say that the above instrument is $A$\textit{-compatible} if $\tr[\Phi_x(\rho)]=\tr[\rho A(x)]$ for all $\rho \in \cS(\cH)$. Note that for every instrument $\cI = \{\Phi_x\}$, there exists a unique measurement $A$, such that $\cI$ is $A$-compatible. Indeed, we have that $\tr[ \Phi_x(\rho) ] = \tr[ \rho \Phi^\ast_x( \Id ) ]$. Thus, defining $A(x) \equiv \Phi^\ast_x( \Id )$, we have that $\tr[ \Phi_x(\rho) ] = \tr[ A(x) \rho ]$, and this $A(x)$ is unique, positive semidefinite and $\sum_x A(x) = \Id$, which follows from the fact that the dual of a CPTP map is a CP unital map. For more results regarding quantum instruments, we refer the reader to Refs.~\cite{Davies,Pellonpaa-1,Pellonpaa-2, Leevi,lever,Gudder-1,Gudder-2,Heino-Coexist,Gudder-3,Chiribella,Holevo,Ozawa-1, Dressel,Ozawa-2,Heino-stro}. 


\subsection{Three kinds of compatibility in quantum theory}\label{subsec:3comp}

One possible definition of compatibility of quantum devices is that they can be performed jointly. That is, a pair of devices is compatible if there exists a joint device, such that applying the joint device reproduces \textit{both} of the outcomes of the compatible devices. If two devices are not compatible, we say that they are \textit{incompatible}. Arguably, the most studied notions of compatibility in quantum theory are the following \cite{heino-review}:

1.~\emph{Measurement compatibility:}
Two measurements $A=\{A(x)\}$ and $B=\{B(y)\}$ are compatible if there exists a measurement $\cG=\{G(x,y)\}$ with outcome set $\Omega_{\cG}=\Omega_A\times \Omega_B$ such that
\begin{equation}
A(x)=\sum_yG(x,y);~ B(y)=\sum_xG(x,y)
\end{equation}
for all $x\in\Omega_A$ and $y\in\Omega_B$. Through measuring $G$, one can simultaneously recover the outputs of both $A$ and $B$. That is, the distribution $p(x,y) \equiv \tr[ G(x,y) \rho ]$ is a joint distribution of $p(x) \equiv \tr[ A(x) \rho ]$ and $p(y) \equiv \tr[ B(y) \rho]$ for all $\rho$.
 
2.~\emph{Channel compatibility:}
Two quantum channels $\Lambda_1:\cS(\cH)\rightarrow\cS(\cK_1)$ and $\Lambda_2:\cS(\cH)\rightarrow\cS(\cK_2)$ are compatible if there exists a quantum channel $\Lambda:\cS(\cH)\rightarrow\cS(\cK_1\otimes\cK_2)$ such that $\Lambda_1(\rho)=\tr_{\cK_2}[\Lambda(\rho)]$ and $\Lambda_2(\rho)=\tr_{\cK_1}[\Lambda(\rho)]$  for all $\rho\in\cS(\cH)$. Through implementing the channel $\Lambda$, one can simultaneously recover the outputs of both $\Lambda_1$ and $\Lambda_2$. That is, $\Lambda(\rho)$ is a joint state of $\Lambda_1(\rho)$ and $\Lambda_2(\rho)$ for all $\rho$.

3.~\emph{Measurement-channel compatibility:}
A measurement $A=\{A(x)\}$ acting on the Hilbert space $\cH$ and a quantum channel $\Lambda:\cS(\cH)\rightarrow\cS(\cK)$ are compatible if there exists a quantum instrument $\cI=\{\Phi_x:\cS(\cH)\rightarrow\cL^+(\cK)\}$ such that $\tr[\Phi_x(\rho)]=\tr[\rho A(x)]$ for all $x\in\Omega_A$ and $\rho\in\cS(\cH)$ and $\sum_x\Phi_x=\Lambda$. Through implementing the quantum instrument $\cI$, one can simultaneously recover the outputs of both $A$ and $\Lambda$.

\section{Compatibility of quantum Instruments}\label{sec:main}

\subsection{Definitions and concepts}\label{subsec:main_1}

While the compatibility of instruments has been studied in far less detail than that of measurements or channels, there is an existing definition in the literature, which we refer to as the \textit{traditional} definition:

\begin{definition}[Traditional compatibility]
Two quantum instruments $\cI_1=\{\Phi^{1}_x:\cS(\cH)\rightarrow\cL^+(\cK)\}$ and $\cI_2=\{\Phi^2_y:\cS(\cH)\rightarrow\cL^+(\cK)\}$ are (traditionally) compatible if there exists an instrument $\cI=\{\Phi_{xy}:\cS(\cH)\rightarrow\cL^+(\cK)\}$ such that $\sum_y\Phi_{xy}=\Phi^{1}_x$ and $\sum_x\Phi_{xy}=\Phi^{2}_y$ for all $x,y$.\label{def:Trad _compa}
\end{definition}

This definition appears in Ref.~\cite[Definition~3]{Heino-stro}, and in
Ref.~\cite[Definition~2.5]{lever}. The same definition is given in
Ref.~\cite[page~15]{Gudder-1}, under the name ``coexistence''. Notice that traditional compatibility can only be defined for instruments with the same quantum output space. Intuitively, the joint instrument $\cI$ in Definition \ref{def:Trad _compa} reproduces both of the classical outputs $x$ and $y$ of $\cI_1$ and $\cI_2$, and by classical post-processing, one can recover \textit{either} one of the two quantum outputs $\Phi^1_x(\rho)$ or $\Phi^2_y(\rho)$. For a schematic representation of the joint instrument, see Fig.~\ref{fig.trad_comp}.

A concept related to traditional compatibility is that of \textit{weak} compatibility.

\begin{definition}[Weak compatibility]
Two quantum instruments $\cI_1=\{\Phi^{1}_x:\cS(\cH)\rightarrow\cL^+(\cK)\}$ and $\cI_2=\{\Phi^2_y:\cS(\cH)\rightarrow\cL^+(\cK)\}$ are weakly compatible if there exists a quantum channel $\Lambda: \cS(\cH) \to \cS(\cK)$ such that $\sum_x\Phi^{1}_{x}=\sum_y\Phi^{2}_{y}=\Lambda$.
\end{definition}

It is known that if a set of instruments is compatible then it is also weakly compatible, but it is easily seen that converse is not true in general \cite{lever}.

Here, we propose a new definition of instrument compatibility, which we refer to as \textit{parallel} compatibility.

\begin{definition}[Parallel compatibility]
Two quantum instruments $\cI_1=\{\Phi^{1}_x:\cS(\cH)\rightarrow\cL^+(\cK_1)\}$ and $\cI_2=\{\Phi^2_y:\cS(\cH)\rightarrow\cL^+(\cK_2)\}$ are parallelly compatible if there exists an instrument  $\cI=\{\Phi_{xy}:\cS(\cH)\rightarrow\cL^+(\cK_1\otimes\cK_2)\}$ such that $\sum_y\tr_{\cK_2}\Phi_{xy}=\Phi^{1}_x$ and $\sum_x\tr_{\cK_1}\Phi_{xy}=\Phi^{2}_y$ for all $x,y$.\label{Def.par_comp}
\end{definition}

Notice that parallel compatibility can be defined for instruments with arbitrary quantum output spaces.  Intuitively, the joint instrument $\cI$ reproduces both of the classical outputs $x$ and $y$ of $\cI_1$ and $\cI_2$, and \textit{both} of the quantum outputs $\Phi^1_x(\rho)$ and $\Phi^2_y(\rho)$ on a tensor product Hilbert space. One can recover these quantum outputs by classical post-processing (i.e., by taking marginal). Furthermore, if  $\cI_1=\{\Phi^{1}_x:\cS(\cH)\rightarrow\cL^+(\cK_1)\}$ and $\cI_2=\{\Phi^2_y:\cS(\cH)\rightarrow\cL^+(\cK_2)\}$ are parallelly compatible with the joint instrument $\cI = \{ \Phi_{xy}: \cS(\cH) \rightarrow \cL^+(\cK_1 \otimes \cK_2) \}$, then the channels $\Phi^1 \equiv \sum_x\Phi^1_x$ and $\Phi^2 \equiv \sum_y\Phi^2_y$ are compatible with the joint channel $\Phi \equiv \sum_{xy}\Phi_{xy}$.

For later convenience, we provide an alternative (but equivalent) definition of parallel compatibility.

\begin{definition}\label{def:parallel_equiv}
Two quantum instruments $\cI_1  =\{ \Phi^1_x : \cS(\cH) \to \cL^+(\cK_1) \}$ and $\cI_2  =\{ \Phi^2_x : \cS(\cH) \to \cL^+(\cK_2) \}$ are parallelly compatible if there exists a quantum instrument $\cI = \{ \Phi_z : \cS(\cH) \to \cL^+(\cK_1 \otimes \cK_2) \}$ such that $\Phi^1_x = \sum_z p_1(x|z) \tr_{\cK_2} \Phi_z$ and $\Phi^2_y = \sum_z p_2(y|z) \tr_{\cK_1} \Phi_z$, where $p_1$ and $p_2$ are conditional probability distributions.
\end{definition}

\begin{proposition}
Definition~\ref{def:parallel_equiv} is equivalent to Definition~\ref{Def.par_comp}. \label{proposi:def-equi}
\end{proposition}

\begin{proof}
The proof is completely analogous to the related proof of equivalent definitions of observable compatibility in \cite[Eqs.~(15)--(17)]{heino-review}.
First, it is clear that Definition \ref{Def.par_comp} is a special case of Definition \ref{def:parallel_equiv}. This can be easily understood by taking $z=(x^{\prime},y^{\prime})$, $p_1(x|z)=\delta_{x^{\prime},x}$ and $p_2(y|z)=\delta_{y^{\prime},y}$ where $\delta_{i,j}$ represents the well-known Kronecker delta function. Hence, we just need to show that if a joint instrument $\cI = \{ \Phi_z : \cS(\cH) \to \cL^+(\cK_1 \otimes \cK_2) \}$ such as the one in Definition \ref{def:parallel_equiv} exists, then there also exists a joint instrument as in Definition \ref{Def.par_comp}. In particular, pick
\begin{equation}
\cI' = \{ \Phi'_{xy} : \cS(\cH) \to \cL^+(\cK_1 \otimes \cK_2)\}
\end{equation}
with $\Phi'_{xy} = \sum_z p_1(x|z)p_2(y|z) \Phi_z$ $\forall x, y$. One can readily check that this is a valid instrument and that $\Phi^1_x = \sum_y \tr_{\cK_2} \Phi'_{xy}$ and $\Phi^2_y = \sum_x \tr_{\cK_1} \Phi'_{xy}$.
\end{proof}

We further illuminate the concept of parallel compatibility through an example and the accompanying figure, Fig.~\ref{fig.par_comp}.
\begin{figure}

\begin{subfigure}{0.5\textwidth}
\begin{center}
\begin{tikzpicture}
\draw [-stealth](2.16,2.5) -- (2.16,1.83);
\draw (2.5,2.18) node {$\rho$};
\draw (2.15,1.55) node {$\cI$};
\draw[black, very thick] (1.7,1.2) rectangle (2.6,1.8);
\draw[black, very thick] (0,0) rectangle (4.3,1.2);
\draw (2,0.7) node {Classical post-processing};
\draw (2,0.5) node {of outcomes};
\draw[black, very thick] (0,0) rectangle (1,-1);
\draw (0.5,-0.5) node {$\cI_1$};
\draw (0.502,-1) -- (0.502,-1.26);
\draw[black, very thick] (3.3,0) rectangle (4.3,-1);
\draw (3.8,-0.5) node {$\cI_2$};
\draw (3.802,-1) -- (3.802,-1.26);
\draw (0.502,-1.26) -- (3.802,-1.26);
\draw [-stealth](2.16,-1.26) -- (2.16,-1.98);

\end{tikzpicture}
\end{center}
\caption{\textbf{Traditional compatibility:} Schematic representation of Definition \ref{def:Trad _compa}. Recovering the quantum output of \textit{either} $\cI_1$ or $\cI_2$ can be done by first implementing the joint instrument $\cI$ on the state $\rho$ and then performing the post-processing of outcomes i.e., taking the marginal over either $x$ or $y$. The downward arrows represent quantum systems. Clearly, in this case there is only one output quantum system.}\label{fig.trad_comp}
\end{subfigure}

\begin{subfigure}{0.5\textwidth}
\begin{center}
\begin{tikzpicture}
\draw [-stealth](2.16,2.06) -- (2.16,1.21);
\draw (2.5,1.7) node {$\rho$};
\draw[black, very thick] (0,0) rectangle (4.3,1.2);
\draw (2,0.7) node {Joint channel $\Lambda$};
\draw (2,0.43) node {(Approx. asymm. cloning)};
\draw[black, very thick] (0,0) rectangle (1,-1);
\draw (0.5,-0.5) node {$\Lambda_1$};
\draw [-stealth](0.502,-1) -- (0.502,-1.7);
\draw [-stealth](-0.5,0) -- (0,0);
\draw[black] (-0.5,-1.34) -- (-0.7,-1.34);
\draw (-0.98,-1.34) node {$\cI_1$};
\draw[black] (-0.5,0) -- (-0.5,-2.7);
\draw [-stealth](-0.5,-2.7) -- (0,-2.7);
\draw[black, very thick] (3.3,0) rectangle (4.3,-1);
\draw (3.8,-0.5) node {$\Lambda_2$};
\draw [-stealth](3.802,-1) -- (3.802,-1.7);
\draw [-stealth](4.9,0) -- (4.3,0);
\draw[black] (4.9,0) -- (4.9,-2.7);
\draw[black] (4.9,-1.34) -- (5.2,-1.34);
\draw (5.43,-1.34) node {$\cI_2$};
\draw [-stealth](4.9,-2.7) -- (4.3,-2.7);
\draw[black, very thick] (0,-1.7) rectangle (1,-2.7);
\draw (0.5,-2.25) node {$\cJ_1$};
\draw [-stealth](0.502,-2.7) -- (0.502,-3.4);
\draw[black, very thick] (3.3,-1.7) rectangle (4.3,-2.7);
\draw (3.8,-2.25) node {$\cJ_2$};
\draw [-stealth](3.802,-2.7) -- (3.802,-3.4);
\end{tikzpicture}
\end{center}
\caption{\textbf{Parallel compatibility: } An example of parallel simultaneous implementation of two instruments (according to Definition \ref{Def.par_comp}), corresponding to Example \ref{example:parallel_compatibility}. The simultaneous implementation of $\cI_1$ and $\cI_2$ can be done through the following steps: (i) implementing the channel $\Lambda$ on the state $\rho$ which is the joint channel of the compatible channels $\Lambda_1$ and $\Lambda_2$ (where $\Lambda_1(\rho)$ and $\Lambda_2(\rho)$ can be considered as the approximate unequal clones (unless $\Lambda_1=\Lambda_2$) of the state $\rho$, in general and therefore, it can be considered as approximate asymmetric cloning), and then (ii) applying the instruments $\cJ_1$ and $\cJ_2$ on $\Lambda_1(\rho)$ and $\Lambda_2(\rho)$ respectively, such that $\cJ_1\circ\Lambda_1=\cI_1$ and $\cJ_2\circ\Lambda_2=\cI_2$. The existence of such a channel $\Lambda$ and such instruments $\cJ_1$ an $\cJ_2$ implies the parallel compatibility of the instruments $\cI_1$ and $\cI_2$, as explained in Example \ref{example:parallel_compatibility}. The downward arrows represent quantum systems. Clearly, in this case there are two output quantum systems.}\label{fig.par_comp}
\end{subfigure}
\caption{Schematic representation of joint instruments for traditionally (Fig.~\ref{fig.trad_comp}) and parallelly (Fig.~\ref{fig.par_comp}) compatible instruments.}
\end{figure}
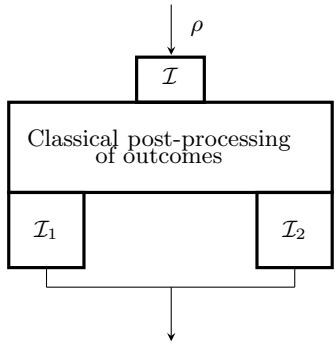
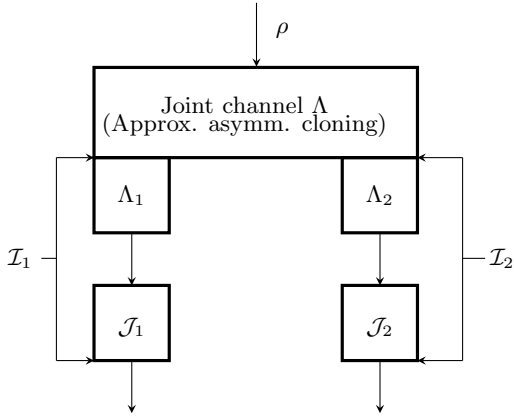

\begin{example}[An example of parallelly compatible instruments]\label{example:parallel_compatibility}
Consider two compatible quantum channels, $\Lambda_1:\cS(\cH)\rightarrow\cS(\cH_1)$ and $\Lambda_2:\cS(\cH)\rightarrow\cS(\cH_2)$ with the joint channel $\Lambda:\cS(\cH)\rightarrow\cS(\cH_1\otimes\cH_2)$. Therefore, since from the no-signaling principle, implementation of local quantum channel on one side does not change the density matrix on the other side, for any arbitrary quantum channel $\Gamma_1:\cS(\cH_1)\rightarrow\cS(\cK_1)$ and $\Gamma_2:\cS(\cH_2)\rightarrow\cS(\cK_2)$, we have that $\tr_{\cK_2}(\Id\otimes\Gamma_2)\circ\Lambda=\Lambda_1$ and  $\tr_{\cK_1}(\Gamma_1\otimes\Id)\circ\Lambda=\Lambda_2$.

 Now consider a pair of arbitrary quantum instruments  $\cJ_1=\{\Phi^{\prime 1}_x:\cS(\cH_1)\rightarrow\cL^+(\cK_1);~\sum_x\Phi^{\prime 1}_x=\Gamma_1\}$ and $\cJ_2=\{\Phi^{\prime 2}_y:\cS(\cH_2)\rightarrow\cL^+(\cK_2);~\sum_y\Phi^{\prime 2}_y=\Gamma_2\}$. Consider another pair of instruments $\cI_1=\cJ_1\circ\Lambda_1=\{\Phi^1_x=\Phi^{\prime 1}_x\circ\Lambda_1:\cS(\cH)\rightarrow\cL^+(\cK_1)\}$ and $\cI_2=\cJ_2\circ\Lambda_2=\{\Phi^2_y=\Phi^{\prime 2}_y\circ\Lambda_2:\cS(\cH)\rightarrow\cL^+(\cK_2)\}$. Then, we show that the instrument $\cI=\{\Phi_{xy}=(\Phi^{\prime 1}_x\otimes\Phi^{\prime 2}_y)\circ\Lambda:\cS(\cH)\rightarrow\cL^+(\cK_1\otimes\cK_2)\}$ is a joint instrument of $\cI_1$ and $\cI_2$.
 
 Clearly, for all $x$
\begin{align}
\Phi^1_x&=\Phi^{\prime 1}_x\circ\Lambda_1\nonumber\\
 &=\Phi^{\prime 1}_x\circ\tr_{\cK_2}(\Id\otimes \Gamma_2)\circ\Lambda\nonumber\\
 &=\Phi^{\prime 1}_x\circ\tr_{\cK_2}(\Id\otimes \sum_y\Phi^{\prime 2}_y)\circ\Lambda\nonumber\\
 &=\tr_{\cK_2}\sum_y(\Phi^{\prime 1}_x\otimes \Phi^{\prime 2}_y)\circ\Lambda\nonumber\\
 &=\sum_{y}\tr_{\cK_2}\Phi_{xy}
 \end{align}
Similarly, $\Phi^2_y=\sum_{x}\tr_{\cK_1}\Phi_{xy}$ for all $x$.  Hence, $\cI_1$ and $\cI_2$ are parallelly compatible with the joint instrument $\cI$.
\end{example}

Next, we show that the notion of traditional compatibility of instruments and that of parallel compatibility of instruments are conceptually different.

\begin{proposition}\label{proposi:para_n_tradi}
There exist pairs of quantum instruments which are parallelly compatible, but not traditionally compatible.
\end{proposition}

\begin{proof}
In Example \ref{example:parallel_compatibility}, $\Gamma_1$ and $\Gamma_2$ can be arbitrary and therefore, $\Gamma_1\circ\Lambda_1$ and $\Gamma_2\circ\Lambda_2$ are not equal in general. Therefore, for the case where $\Gamma_1\circ\Lambda_1\neq \Gamma_2\circ\Lambda_2$, $\cI_1$ and $\cI_2$ are parallelly compatible, but not weakly compatible and therefore not traditionally compatible.
\end{proof}

\begin{proposition}\label{proposi:tradi_n_para}
There exist pairs of quantum instruments which are traditionally compatible, but not parallelly compatible.
\end{proposition}

\begin{proof}
Consider two quantum instruments, $\cI^p=\{\Phi^p_1=p_1\mathfrak{I}, \Phi^p_2=p_2\mathfrak{I}\}$ and $\cI^q=\{\Phi^q_1=q_1\mathfrak{I}, \Phi^q_2=q_2\mathfrak{I}\}$ where $\mathfrak{I}$ is the identity channel and $p_i=\sum_{j}r_{ij}$ and $q_j=\sum_{i}r_{ij}$ for some $\{r_{ij}\geq 0\}_{ij=\{1,2\}}$ with $\sum_{ij}r_{ij}=1$. Clearly, $\cI^p$ and $\cI^q$ are traditionally compatible with the joint instrument $\cI^r=\{r_{ij}\mathfrak{I}\}_{ij=\{1,2\}}$. However, as discussed earlier, if $\cI^p$ and $\cI^q$ are parallelly compatible then $\Phi^p=\sum_i \Phi^p_i=\mathfrak{I}$ and $\Phi^q=\sum_j \Phi^q_j=\mathfrak{I}$ are compatible. But since the identity channel $\mathfrak{I}$ is not compatible with itself (due to the no-cloning theorem), $\Phi^p$ and $\Phi^q$ cannot be compatible, and therefore $\cI^p$ and $\cI^q$ cannot be parallelly compatible.
\end{proof}




\subsection{Arguments against traditional compatibility}\label{subsec:main_2}

In the previous section, we have introduced two notions of instrument compatibility and showed that these notions are conceptually different (neither of them implies the other). Here, we argue that the traditional notion has significant drawbacks.

Let us recall from Section \ref{Sec:preli} that measurements are devices with a quantum input and a classical output, while channels are devices with a quantum input and a quantum output. Furthermore, we say that a pair of such devices is compatible if there exists a joint device that upon taking a quantum input, reproduces \textit{both} of the outputs of the original devices. For measurements, this means that the joint measurement produces a classical output that is the joint measurement outcome of the two compatible measurements. For channels, this means that the joint channel produces a quantum output that is the joint state of the outputs of the compatible channels. According to this principle, when one is looking for a definition of compatibility of instruments, one should look for a joint instrument that reproduces \emph{both} the joint classical and the joint quantum output of the compatible instruments.

It is clear from Definition \ref{def:Trad _compa} that the traditional notion of instrument compatibility provides a joint instrument with \emph{a single quantum output}. Thus, by design, the traditional definition does not allow for producing a \textit{simultaneous} quantum output of \emph{both} of the compatible quantum instruments. Furthermore, this definition only applies to instruments with the same output Hilbert space. Note that for traditionally compatible instruments, one can only recover a single quantum output via classical post-processing. This is not the case for parallel compatibility, where the joint instrument produces a joint state, whose marginals coincide with the quantum outputs of the compatible instruments. Indeed, after performing the joint instrument, one has access to \emph{both} of the quantum outputs, and one can perform further operations on both of them simultaneously. To illuminate this argument, we recall Proposition \ref{proposi:tradi_n_para}, which shows that two ``identity instruments'' (instruments with all their channels being proportional to the identity channel) are traditionally compatible. In any notion of compatibility for which the joint instrument recovers both of the quantum outputs, this would mean that one can recover two copies of the output of an identity channel, which is in clear contradiction with the no-cloning theorem. Thus, we argue that the traditional notion of instrument compatibility does not capture compatibility in the same way as the well-established notions of measurement and channel compatibility do.

As further justification for our argument, the following two propositions show that while traditional compatibility of instruments captures measurement compatibility, it can never capture channel compatibility.
\begin{proposition}
Two measurements $A$ and $B$ are compatible if and only if there exist an $A$-compatible instrument $\cI_A$ and a $B$-compatible instrument $\cI_B$ such that $\cI_A$ and $\cI_B$ are traditionally compatible.
\end{proposition}

\begin{proof}
Suppose that there exists an $A$-compatible instrument $\cI_A=\{\Phi^{A}_x:\cS(\cH)\rightarrow\cL^+(\cK)\}$ and a $B$-compatible instrument $\cI_B=\{\Phi^B_y:\cS(\cH)\rightarrow\cL^+(\cK)\}$ such that $\cI_A$ and $\cI_B$ are traditionally compatible. Then there exists an instrument $\cI=\{\Phi_{xy}:\cS(\cH)\rightarrow\cL^+(\cK)\}$ such that  $\sum_y\Phi_{xy}=\Phi^{1}_x$ and $\sum_x\Phi_{xy}=\Phi^{2}_y$ for all $x,y$. Let us choose a measurement $G=\{G(x,y)\}$ such that $\cI$ is $G$-compatible (recall that such a measurement always exists). Then for any $\rho\in\cS(\cH)$, we have that $\tr[\rho G(x,y)]=\tr[\Phi_{xy}(\rho)]$ for all $x,y$ and for all $\rho\in\cS(\cH)$. Therefore, $\sum_y\tr[\rho G(x,y)]=\tr[\Phi^A_{x}(\rho)]=\tr[\rho A(x)]$ and $\sum_x\tr[\rho G(x,y)]=\tr[\Phi^B_{xy}(\rho)]=\tr[\rho B(y)]$ for all $x,y$ and for all $\rho\in\cS(\cH)$. Therefore, $A$ and $B$ are compatible with the joint measurement $G=\{G(x,y)\}$.

Now suppose that $A$ and $B$ are compatible with the joint measurement $G=\{G(x,y)\}$, and let $\cI=\{\Phi_{xy}\}$ be a $G$-compatible instrument [e.g., choose the L\"uders instrument, $\Phi_{xy}(\rho) = \sqrt{G(x,y)} \rho \sqrt{G(x,y)}$]. Then it is easy to check that the instrument $\cI^A \equiv \{\Phi^A_{x}=\sum_y\Phi_{xy}\}$ is an $A$-compatible instrument and $\cI^B \equiv \{\Phi^B_{x}=\sum_y\Phi_{xy}\}$ is a $B$-compatible instrument.
\end{proof}

\begin{proposition}
The traditional compatibility of quantum instruments cannot capture the compatibility of quantum channels.
\end{proposition}

\begin{proof} 
Take two instruments, $\cI_A=\{\Phi^{A}_x:\cS(\cH)\rightarrow\cL^+(\cK)\}$ and $\cI_B=\{\Phi^B_y:\cS(\cH)\rightarrow\cL^+(\cK)\}$ such that they are traditionally compatible. Taking $\Phi^A \equiv \sum_x\Phi^{A}_x$ and $\Phi^B = \sum_y\Phi^B_y$, we recall that traditional compatibility implies $\Phi^A=\Phi^B$. Therefore, traditional compatibility cannot capture channel compatibility in full generality, since there exist compatible channels that are not equal.
\end{proof}

\subsection{Arguments for parallel compatibility}\label{subsec:main_3}

In this section, we argue that parallel compatibility does not have the flaws of traditional compatibility. In the previous section, we already argued for this from the conceptual viewpoint---that is, parallel compatibility allows for the simultaneous recovery of both of the quantum outputs of the compatible instruments. Here, we further justify the adequacy of parallel compatibility by showing that this notion captures the idea of measurement compatibility, channel compatibility and measurement-channel compatibility. We summarise these findings in the following theorem.

\begin{theorem}
Parallel compatibility of instruments captures measurement compatibility, measurement compatibility and measurement-channel compatibility.

1.~Two measurements, $A$ and $B$, are compatible if and only if there exist an $A$-compatible instrument $\cI_A$ and a $B$-compatible instrument $\cI_B$ such that $\cI_A$ and $\cI_B$ are parallelly compatible.

2.~Two quantum channels, $\Phi^1:\cS(\cH)\rightarrow\cS(\cK_1)$ and $\Phi^2:\cS(\cH)\rightarrow\cS(\cK_2)$, are compatible if and only if there exist two parallelly compatible instruments $\cI_1=\{\Phi^{1}_x:\cS(\cH)\rightarrow\cL^+(\cK_1)\}$ and $\cI_2=\{\Phi^2_y:\cS(\cH)\rightarrow\cL^+(\cK_2)\}$ such that $\sum_x\Phi^{1}_x=\Phi^1$ and $\sum_y\Phi^2_y=\Phi^2$.

3.~If an $A$-compatible instrument $\cI_A=\{\Phi^{A}_x:\cS(\cH)\rightarrow\cL^+(\cK_1)\}$ and a $B$-compatible instrument $\cI_B=\{\Phi^B_y:\cS(\cH)\rightarrow\cL^+(\cK_2)\}$, such that $\sum_x\Phi^{A}_x=\Phi^A$ and $\sum_y\Phi^B_y=\Phi^B$, are parallelly compatible, then $A$ and $B$ are both compatible with both $\Phi^A$  and $\Phi^B$.
\label{th:u}
\end{theorem}

\begin{proof}
We start with the ``if'' part of statement 1. Suppose that there exists an $A$-compatible quantum instrument $\cI_A=\{\Phi^{1}_x:\cS(\cH)\rightarrow\cL^+(\cK_1)\}$ and  a $B$-compatible quantum instrument $\cI_B=\{\Phi^2_y:\cS(\cH)\rightarrow\cL^+(\cK_2)\}$ such that they are parallelly compatible. Then there exists an instrument  $\cI=\{\Phi_{xy}:\cS(\cH)\rightarrow\cL^+(\cK_1\otimes\cK_2)\}$ such that $\sum_y\tr_{\cK_2}\Phi_{xy}=\Phi^{1}_x$ and $\sum_x\tr_{\cK_1}\Phi_{xy}=\Phi^{2}_y$ for all $x,y$. Since $\cI_A$ is $A$-compatible, we have that $\tr[\rho A(x)]=\tr[\Phi^{1}_x(\rho)]=\tr[\rho(\Phi^{1}_x)^\ast (\Id_{\cK_1})]$ for all $\rho \in \cS(\cH)$ and all $x$. This implies
\begin{equation}
\sum_y\Phi_{xy}^\ast (\Id_{\cK_1\otimes\cK_2})=(\Phi^{1}_x)^\ast (\Id_{\cK_1})=A(x) \quad \forall x,
\end{equation}
where the first equality is a consequence of Observation \ref{obs:marginal_dual}. Similarly, we have that $\sum_x\Phi_{xy}^*(\Id_{\cK_1\otimes\cK_2})=(\Phi^{2}_y)^*(\Id_{\cK_2})=B(y)$ for all $y$.

Now we define the measurement $G=\{G(x,y)=\Phi_{xy}^*(\Id_{\cK_1\otimes\cK_2})\}$ (which is the unique measurement compatible with $\cI$).
For this measurement, we have that $\sum_{y}G(x,y)=\sum_y\Phi_{xy}^*(\Id_{\cK_1\otimes\cK_2})=(\Phi^{1}_x)^*(\Id_{\cK_1})=A(x)$ and $\sum_{x}G(x,y)=\sum_x\Phi_{xy}^*(\Id_{\cK_1\otimes\cK_2})=(\Phi^{1}_y)^*(\Id_{\cK_2})=B(y)$ and therefore $A$ and $B$ are compatible via the joint measurement $G$.

Now we move on to the ``only if'' part. Let $\{A(x)\}$ and $\{B(y)\}$ be compatible measurements on the Hilbert space $\cH$.
Let $\{G(x,y)\}$ denote a joint measurement for $A$ and $B$, and consider the Naimark dilation $\{\Pi(x,y)\}$ on the Hilbert space $\cK \equiv \cH \otimes \cH'$. That is, for every state $\rho \in \cS(\cH)$ we have that $\tr[ G(x,y) \rho ] = \tr[ \Pi(x,y) ( \rho \otimes \ketbraq{0} )]$ for some fixed state $\ket{0}$ on $\cH'$, and $\{\Pi(x,y)\}$ is a projection-valued measure (PVM), i.e., $\Pi^2(x,y) = \Pi(x,y)$ for all $x,y$. Furthermore, consider a rank-1 ``fine-graining'', $\tilde{\Pi}(z) = \ketbraq{\phi_z}$, of $\Pi(x,y)$, i.e., a rank-1 projective measurement such that
\begin{equation}
\Pi(x,y) = \sum_{z \in P(x,y)} \tilde{\Pi}(z),
\end{equation}
where P(x,y) is the subset of all the possible values of z such that $\tilde{\Pi}(z)$ is in the support of $\Pi(x,y)$.

Consider the instrument
\begin{equation}
\cI_{\tilde{\Pi}} = \{ \Phi^{\tilde{\Pi}}_{z}: \cS(\cH) \to \cL^+(\cK) ~|~ \Phi^{\tilde{\Pi}}_{z}(\rho) = \tilde{\Pi}(z) (\rho \otimes \ketbraq{0}) \tilde{\Pi}(z) \}
\end{equation}
defined by the channels mapping $\rho$ to the (un-normalised) post-measurement state of the dilated and fine-grained measurement $\tilde{\Pi}$. It is clear that $\Phi^{\tilde{\Pi}}_{z}$ is CP with the single Kraus operator $\tilde{\Pi}(z)( \I_\cH \otimes \ket{0}_{\cH'})$. Further, it is also clear that $\Phi^{\tilde{\Pi}} \equiv \sum_{z} \Phi^{\tilde{\Pi}}_{z}$ is CPTP, since
\begin{align}
\tr[ \Phi^{\tilde{\Pi}}( \rho ) ] &= \sum_{z} \tr[ \tilde{\Pi}(z)( \rho \otimes \ketbraq{0} ) \tilde{\Pi}(z) ]\nonumber\\
 &= \tr[ \sum_{z} \tilde{\Pi}(z)( \rho \otimes \ketbraq{0} )]\nonumber\\
 & = \tr[ \rho \otimes \ketbraq{0} ] \nonumber\\
 & = \tr \rho
\end{align}
for all $\rho$.

Since $\tilde{\Pi}$ is a PVM, its unnormalised post-measurement states, $\tilde{\rho}_{z} \equiv \tilde{\Pi}(z) (\rho \otimes \ketbraq{0}) \tilde{\Pi}(z) \in \cL( \cK )$, are rank-1 and pairwise orthogonal. Explicitly, they are given by
\begin{equation}
\tilde{\rho}_{z} = \tilde{\Pi}(z) (\rho \otimes \ketbraq{0}) \tilde{\Pi}(z) = \lambda_z(\rho) \ketbraq{\phi_z},
\end{equation}
where $\lambda_z(\rho) = \tr[ \tilde{\Pi}(z) (\rho \otimes \ketbraq{0})]$. Consider then an isometry $V : \cK \to \cK_1 \otimes \cK_2$ (with $\cK_1 \cong \cK_2 \cong \cK$) such that
\begin{equation}
V \ket{\phi_z}_{\cK} = \ket{\phi_z}_{\cK_1} \otimes \ket{\phi_z}_{\cK_2} \quad \forall z,
\end{equation}
which always exists, since one can always clone a set of fixed orthogonal states. Hence,
\begin{equation}
V \tilde{\rho}_z V^\dagger = \lambda_z(\rho) \ketbraq{\phi_z}_{\cK_1} \otimes \ketbraq{\phi_z}_{\cK_2}.
\end{equation}
Let us then define the instrument
\begin{align}
\cI = \{ \Phi_{z} : \cS(\cH) \to \cL^+( \cK_1 \otimes \cK_2 )\}
\end{align}
with $\Phi_{z}(\rho) = V \tilde{\Pi}(z) (\rho \otimes \ketbraq{0}) \tilde{\Pi}(z) V^\dagger = \lambda_z(\rho) \ketbraq{\phi_z}_{\cK_1} \otimes \ketbraq{\phi_z}_{\cK_2}$. Since the $\Phi_{z}$ are just the composition of $\Phi^{\tilde{\Pi}}_{z}$ with the isometry $V$, it is clear that these are also CP maps and that $\Phi \equiv \sum_{z} \Phi_{z}$ is a CPTP map, and hence $\cI$ is a valid instrument.

Let us now define the instruments
\begin{align}
\cI_A & \equiv \{ \Phi^A_x : \cS(\cH) \to \cL^+(\cK_1)\}
\end{align}
where $\Phi^A_x(\rho) = \sum_y \sum_{z \in P(x,y)} \tr_{ \cK_2 } \Phi_{z} (\rho) = \sum_y \sum_{z \in P(x,y)} \lambda_z(\rho) \ketbraq{\phi_z}_{\cK_1}$ and
\begin{align}
\cI_B &  \equiv \{ \Phi^B_y : \cS(\cH) \to \cL^+(\cK_2)\}
\end{align}
where $\Phi^B_y(\rho) = \sum_x \sum_{z \in P(x,y)} \tr_{ \cK_1 } \Phi_{z} (\rho) = \sum_x \sum_{z \in P(x,y)} \lambda_z(\rho) \ketbraq{\phi_z}_{\cK_2} \}$.
It is clear that these are valid instruments, and by definition, $\cI_A$ and $\cI_B$ are parallelly compatible with the joint instrument $\cI$. Furthermore, we have that
\begin{align}
\tr[ \Phi^A_x( \rho ) ] & = \tr[ \sum_y \sum_{z \in P(x,y)} \lambda_z(\rho) \ketbraq{\phi_z}_{\cK_1} ] \nonumber\\
&=\sum_y \sum_{z \in P(x,y)} \lambda_z(\rho)\nonumber\\
&=\sum_y \sum_{z \in P(x,y)} \tr[ \tilde{\Pi}(z)( \rho \otimes \ketbraq{0} ) ]\nonumber\\
&=\sum_y \tr[ \Pi(x,y)( \rho \otimes \ketbraq{0} ) ]\nonumber\\
&=\sum_y \tr[ G(x,y) \rho ]\nonumber\\
&= \tr[ A(x) \rho ],
\end{align}
that is, $\cI_A$ is $A$-compatible, and similarly we have that $\cI_B$ is $B$-compatible. This finishes the proof of statement~1.

We continue with the proof of statement 2. To show the ``only if'' part, notice that a quantum channel $\Lambda$ can be considered as a single-outcome quantum instrument $\cI_{\Lambda}$. Therefore, the compatibility of two quantum channels $\Phi^1$ and $\Phi^2$ implies the parallel compatibility of the two instruments $\cI_{1} \equiv \{\Phi^1\}$ and $\cI_{2} \equiv \{\Phi^2\}$. The ``if'' part follows straightforwardly from the definition of parallel compatibility, and it is already explained below Definition \ref{Def.par_comp}.

Last, we prove statement 3. By definition, $A$ is compatible with $\Phi^A$ and $B$ is compatible with $\Phi^B$. Since $\cI_A$ and $\cI_B$ are parallelly compatible, there exists a quantum instrument $\cI =\{\Phi_{xy}:\cS(\cH)\rightarrow\cL^+(\cK_1\otimes\cK_2)\}$ such that $\Phi^{A}_x=\sum_y\tr_{\cK_2}\Phi_{xy}$ and $\Phi^{B}_y=\sum_x\tr_{\cK_1}\Phi_{xy}$. Since $\cI_A$ is an $A$-compatible instrument, we have that for all $\rho\in\cS(\cH)$ and $x\in\Omega_A$,
\begin{align}
\tr_{\cH}[\rho A(x)]&=\tr_{\cK_1}[\Phi^A_x(\rho)]\nonumber\\
&=\tr_{\cK_1}[\sum_y\tr_{\cK_2}[\Phi_{xy}(\rho)]\nonumber\\
&=\sum_y[\tr_{\cK_1}\tr_{\cK_2}[\Phi_{xy}(\rho)]\nonumber\\
&=\tr_{\cK_2}[\sum_y[\tr_{\cK_1}[\Phi_{xy}(\rho)]]\nonumber\\
&=\tr_{\cK_2}[\Phi^{\prime A}_{x}(\rho)],\nonumber\\
\end{align}
where $\Phi^{\prime A}_{x} \equiv \sum_y \tr_{\cK_1} \Phi_{xy}$. Clearly, $\Phi^{\prime A}_x:\cS(\cH)\rightarrow\cL^+(\cK_2)$ for all $x\in\Omega_A$, and
\begin{align}
\sum_x\Phi^{\prime A}_{x}&=\sum_x\sum_y\tr_{\cK_1}\Phi_{xy}\nonumber\\
&=\sum_y\Phi^B_y\nonumber\\
&=\Phi^B,
\end{align}
and thus, $\cI^{\prime}_A \equiv \{\Phi^{\prime A}_x:\cS(\cH)\rightarrow\cL^+(\cK_2)\}$ is a quantum instrument. Hence, $A$ and $\Phi^B$ are compatible through the instrument $\cI^{\prime}_A$. Similarly one can prove that $B$ and $\Phi^A$ are compatible as well.
\end{proof}

Thus, we have proved that parallel compatibility captures the three kinds of compatibilities between basic quantum devices, that is, measurement compatibility, channel compatibility and measurement-channel compatibility.

Last, we point out that while every instrument has a unique measurement it is compatible with, the converse is not true: there are multiple instruments that are compatible with the same measurements (corresponding to different implementations of the same measurement). One consequence of this is that while for every pair of compatible measurements $A$ and $B$ there exists an $A$-compatible instrument $\cI_A$ and a $B$-compatible instrument $\cI_B$ such that they are parallelly compatible (statement 1.~of Theorem \ref{th:u}), not \textit{every} $A$-compatible and $B$-compatible instrument will be parallelly compatible. Even for a single measurement, two different instruments that are compatible with it, may not be parallelly compatible, as the following example shows:

\begin{example}[Two parallelly incompatible instruments associated with the same measurement]
A trivial measurement $J=\{J(x)=p_x\Id\}$ is compatible with any quantum channel $\Lambda$ through the instrument $\cI_{J,\Lambda}=\{p_x\Lambda\}$ \cite{heino-review}. Let us consider a channel $\Gamma$, which is incompatible with $\Lambda$. Clearly, $J$ is also compatible with the quantum channel $\Gamma$ through the instrument $\cI_{J,\Gamma}=\{p_x\Gamma\}$. From Theorem \ref{th:u}, we know that if two instruments are parallelly compatible, then their corresponding channels are compatible. Then, since $\Gamma$ and $\Lambda$ were chosen to be incompatible, the instruments $\cI_{J,\Gamma}$ and $\cI_{J,\Lambda}$ cannot be parallelly compatible.
\end{example} 

\section{Conclusion}\label{sec:conc}

In this paper, we introduced the concept of parallel compatibility of instruments and showed that this concept is different from the traditional definition of instrument compatibility. We argued that the traditional definition of compatibility of instruments is conceptually incomplete, and provided arguments for the adequacy of parallel compatibility. We showed that the definition of parallel compatibility of quantum instruments can capture the idea of measurement compatibility, channel compatibility and measurement-channel compatibility.

The notion of parallel compatibility may be relevant to various information theoretic tasks. First, suppose that Charlie wants to simultaneously transfer information to two parties, Alice and Bob, separated by a long distance. The information is transmitted through a quantum state, and Alice is retrieving the information through a measurement $A$, while Bob is retrieving the information through a measurement $B$. Then, if $A$ and $B$ are compatible, Fig.~\ref{fig.par_comp} suggests that the transmission can be done via the joint instrument of the corresponding instruments $\cI_A$ and $\cI_B$.
Second, consider the same scenario in a cryptographic setting, where Charlie and Alice aim to perform a key distribution task, and Bob is an eavesdropper. Then, Fig.~\ref{fig.par_comp} suggests that cloning-type attacks can be modelled through parallel compatibility of quantum instruments. We leave the exploration of the role of parallel compatibility in such information theoretic tasks, as well as the further characterisation of compatible instruments for future work.

\section{Acknowledgements}
The authors thank Prof.~Sibasish Ghosh and Erkka Haapasalo for fruitful discussions and their valuable comments on this work. MF acknowledges funding from the Government of Spain (FIS2020-TRANQI and Severo Ochoa CEX2019-000910-S), Fundació Cellex, Fundació Mir-Puig and Generalitat de Catalunya (CERCA, AGAUR SGR 1381). This project has received funding from the European Union’s Horizon 2020 research and innovation programme under the Marie Sk\l odowska-Curie grant agreement No 847517.

\end{document}